\documentclass[12pt,reqno]{amsart}

\newcommand\version{September 18, 2014}

\usepackage{amsmath,amsfonts,amsthm,amssymb,amsxtra}

\newtheorem{theorem}{Theorem}[section]
\newtheorem{proposition}[theorem]{Proposition}
\newtheorem{lemma}[theorem]{Lemma}

\theoremstyle{definition}

\theoremstyle{remark}

\newtheorem{remark}[theorem]{Remark}

\numberwithin{equation}{section}

\renewcommand{\epsilon}{\varepsilon}

\newcommand{\loc}{{\rm loc}}

\newcommand{\N}{\mathbb{N}}

\renewcommand{\phi}{\varphi}
\newcommand{\R}{\mathbb{R}}

\DeclareMathOperator{\spec}{spec}

\DeclareMathOperator{\tr}{Tr}

\begin{document}

\title[Large polaron systems --- \version]{Ground state energy of large polaron systems}

\author{Rafael D. Benguria}
\address{Rafael D. Benguria, Instituto de F\'\i sica, Pontificia Universidad Cat\'olica de Chile, Santiago, Chile}
\email{rbenguri@fis.puc.cl}

\author{Rupert L. Frank}
\address{Rupert L. Frank, Mathematics 253-37, Caltech, Pasadena, CA 91125, USA}
\email{rlfrank@caltech.edu}

\author{Elliott H. Lieb}
\address{Elliott H. Lieb, Departments of Mathematics and Physics, Princeton
University, Princeton, NJ 08544, USA}
\email{lieb@princeton.edu}

\begin{abstract}
The last unsolved problem about the many-polaron system, in the
Pekar--Tomasevich approximation, is the case of bosons with the
electron-electron Coulomb repulsion of strength exactly 1 (the 'neutral
case'). 
We prove that the ground state energy, for large $N$, goes exactly as
$-N^{7/5}$, and we give upper and lower bounds on the asymptotic coefficient that agree to 
within a factor of $2^{2/5}$. 
\end{abstract}

\maketitle

\centerline{ \version}
\renewcommand{\thefootnote}{${}$}\footnotetext{\copyright\, 2014 by
  the authors. This paper may be reproduced, in its entirety, for
  non-commercial purposes.\\
Work partially supported by Fondecyt (Chile) project 112--0836 and the Iniciativa Cient\'\i fica Milenio (Chile) through the Millenium Nucleus RC--120002 ``F\'\i sica Matem\'atica''  (R.D.B.), and NSF
grants PHY--1347399 and DMS--1363432 (R.L.F.), PHY--0965859 and PHY--1265118 (E.H.L.)}

\section{Introduction and main results}

In this paper we are concerned with the ground state energy of a system of
$N$ polarons in the Pekar--Tomasevich approximation \cite{PT},
which is  derived from the Fr\"ohlich polaron \cite{F}, in the limit of
large  coupling constant $\alpha$ (see \eqref{froh}). The Pekar--Tomasevich
energy functional is
\begin{equation}
\label{eq:energy}
\mathcal E^{(N)}_U[\psi] = \int_{\R^{3N}} \left( \sum_{j=1}^N |\nabla_j \psi|^2 + U \sum_{j<k} \frac{|\psi|^2}{|x_j-x_k|} \right) dx - D(\rho_\psi,\rho_\psi)
\end{equation}
for $\psi\in H^1(\R^{3N})$ with $\int_{\R^{3N}} |\psi|^2\,dx=1$. (We will
write $\Vert \psi\Vert $ to denote the $L^2$, not the $H^1$ norm.)
We have used the usual notations
$$
\rho_\psi(x) = \sum_{j=1}^N \int\!\cdots\!\int_{\R^{3(N-1)}} |\psi(x_1,\ldots,x_{j-1},x,x_{j+1},\ldots,x_N)|^2 \,dx_1\ldots dx_{j-1} dx_{j+1}\ldots dx_N
$$
for the particle density corresponding to $\psi$ and
$$
D(\rho,\sigma)= \frac12 \iint_{\R^3\times\R^3} \frac{\overline{\rho(x)}\ \sigma(x')}{|x-x'|} \,dx\,dx'
$$
for the Coulomb energy. The dimensionless parameter $U>0$ in \eqref{eq:energy} describes the strength of the Coulomb repulsion between the particles relative to the strength of their self-attraction. (Originally, there is another parameter $\alpha>0$ in front of $D(\rho_\psi,\rho_\psi)$, but by scaling we may assume that $\alpha=1$.)

We are concerned both with the case of bosonic and of fermionic statistics and we denote the corresponding ground state energies by
$$
E_U^{(b)}(N) = \inf\left\{ \mathcal E^{(N)}_U[\psi]:\ \mathrm{symmetric}\ \psi\in H^1(\R^{3N})\,,\ \int_{\R^{3N}} |\psi|^2 \,dx = 1 \right\}
$$
and
$$
E_U^{(f)}(N) = \inf\left\{ \mathcal E^{(N)}_U[\psi]:\
\mathrm{antisymmetric}\ \psi\in H^1(\R^{3N})\,,\ \int_{\R^{3N}} |\psi|^2
\,dx = 1 \right\} \,.
$$
For simplicity, we ignore spin. It is well known that $E_U^{(b)}(N)$
coincides with the infimum of $\mathcal E^{(N)}_U[\psi]$ over all $\psi\in
H^1(\R^{3N})$ with $\int_{\R^{3N}}|\psi|^2\,dx =1$, that is, the assumption
`symmetric' in the definition of $E_U^{(b)}(N)$ can be dropped. This
implies, in particular, that $E_U^{(b)}(N)\leq E_U^{(f)}(N)$ for all
$N\in\N$.

Let us review what is known about the large $N$ behavior of $E^{(b)}_U(N)$ and $E^{(f)}_U(N)$. These results depend crucially on the sign of $U-1$. For $U<1$ and fermions it is shown in \cite{GM} that
\begin{equation}\label{fer}
-\infty< \liminf_{N\to\infty} N^{-7/3} E^{(f)}(N) \leq \limsup_{N\to\infty} N^{-7/3} E^{(f)}(N) \leq e_U^{(f)} \,,
\end{equation}
for some explicit constant $e_U^{(f)}$ (defined in \eqref{little} below).
We shall prove that both $\leq$ in \eqref{fer} are $=$, in fact. 
For $U<1$ and bosons it was noted in \cite{FLST} that $ \liminf_{N\to\infty}
N^{-3} E^{(b)}(N)$ and $\liminf_{N\to\infty} N^{-3} E^{(b)}(N)$ are finite
and in \cite{BB} it was shown that
$$
\lim_{N\to\infty} N^{-3} E^{(b)}(N) = e_U^{(b)}
$$
for some explicit constant $e_U^{(b)}$. Note that in these cases the
thermodynamic limit does not exist. 

The situation changes when $U>1$. In this case it was shown in \cite{FLST}
that 
$\liminf_{N\to\infty} N^{-1} E^{(f)}_U(N) \geq \liminf_{N\to\infty}
N^{-1} E^{(b)}_U(N)>-\infty$ and it was deduced that
$$
\lim_{N\to\infty} N^{-1} E^{(b)}_U(N)
\qquad\text{and}\qquad
\lim_{N\to\infty} N^{-1} E^{(f)}_U(N)
\qquad\text{exist}.
$$
In the critical case $U=1$, (also known as
the neutral case) for fermions, it is shown in \cite{GM} that
$\liminf_{N\to\infty} N^{-1} E^{(f)}_U(N)>-\infty$. By the same
sub-additivity argument as in \cite{FLST} this implies that
$$
\lim_{N\to\infty} N^{-1} E^{(f)}_1(N)
\qquad\text{exists}.
$$
Thus, our understanding of polaron ground state energies is complete except for the bosonic case with $U=1$. Our goal in this paper is to fill this gap. The following is our main result.

\begin{theorem}\label{main}
In the bosonic case with $U=1$,
\begin{equation}
\label{eq:main}
-2^{2/5} A \leq \liminf_{N\to\infty} N^{-7/5} E^{(b)}_1(N) \leq 
\limsup_{N\to\infty} N^{-7/5} E^{(b)}_1(N) \leq -A \,,
\end{equation}
where
\begin{equation}
\label{eq:a}
-A = \inf\left\{ \int_{\R^3} |\nabla\phi|^2\,dx - I_0 \int_{\R^3} |\phi|^{5/2} \,dx :\ \phi\in H^1(\R^3)\,,\ \int_{\R^3} |\phi|^2\,dx = 1\right\} \,.
\end{equation}
with
\begin{equation}
\label{eq:i0}
I_0 = \frac 25 \left(\frac{2}{\pi}\right)^{1/4}
\frac{\Gamma(3/4)}{\Gamma(5/4)} 
\simeq 0.60868\,.
\end{equation}
\end{theorem}

We emphasize that \eqref{eq:main} identifies the correct growth rate of $E_1^{(b)}(N)$ as $N\to\infty$. Our asymptotic upper and lower bound, however, differ by a factor of $2^{2/5}$. We believe that the upper bound is the correct one. Our proof of the theorem is constructive and leads to explicit error bounds. For instance, for the upper bound, we obtain
\begin{equation}
\label{eq:mainupper}
E^{(b)}_1(N) \leq -A N^{7/5} (1- CN^{-1/35}) \,.
\end{equation}

\begin{remark}
Consider the Fr\"ohlich Hamiltonian \cite{F},
\begin{equation}\label{froh}
H_{U,\alpha}^{(N)} = \sum_{j=1}^N \left(-\Delta_j + \sqrt\alpha \phi(x_j)\right) + U \sum_{i<j} |x_i-x_j|^{-1} + \int_{\R^3} a_k^* a_k \,dk
\end{equation}
in $L^2_{\mathrm{symm}}(\R^{3N}) \otimes\mathcal F(L^2(\R^3))$, where
$\mathcal F$ denotes the bosonic Fock space and where 
$$
\phi(x) = \frac{1}{2\pi} \int_{\R^3} \left( a_k e^{ik\cdot x} + a_k^* e^{-ik\cdot x} \right) \frac{dk}{|k|} \,,
$$
with annihilation and creation operators $a_k$ and $a_k^*$.
Since a (rescaled) Pekar functional \cite{P} is an upper bound on $\inf\spec
H_{U,\alpha}^{(N)}$ \cite{FLST}, we conclude from Theorem \ref{main} that
for $U=\alpha$,
$$
\limsup_{N\to\infty} N^{-7/3} \inf\spec H_{\alpha,\alpha}^{(N)} \leq - A \alpha^2 \,.
$$
In particular, the thermodynamic limit does not exist. If the particles are
treated as fermions, the existence of such a limit  is still an open
problem.  {\it End of Remark}
\end{remark}

In our second theorem, proved in the appendix, we show that the upper bound obtained in \cite{GM} in the fermionic case for $U<1$ is, in fact, asymptotically correct.

\begin{theorem}\label{main2}
In the fermionic case with $0<U<1$,
\begin{equation}
\label{eq:main2}
\lim_{N\to\infty} N^{-7/3} E^{(f)}_U(N) = e_U^{(f)} \,,
\end{equation}
with
\begin{equation}\label{little}
e_U^{(f)} = \inf\left\{ \frac{3}{5} (6\pi^2)^{2/3} \int_{\R^3} \rho^{5/3}\,dx - (1-U) D(\rho,\rho):\ \rho\geq 0\,,\ \int_{\R^3} \rho\,dx =1 \right\} \,.
\end{equation}
\end{theorem}

We emphasize that the proof of Theorem \ref{main} is much more involved than
that of Theorem \ref{main2}, which we include mostly for the sake of
completeness.

\subsection*{Acknowledgement}
The authors are grateful to Jan Philip Solovej for useful discussions about Theorem \ref{main}.


\section{Proof of Theorem \ref{main}}

\subsection{The lower bound}

Oddly enough, the {\it lower} bound is the easy one for us because we have
available the
results in \cite{LiSo2} for  the two-component charged Bose gas.
With its use we can deduce our
lower bound from an asymptotic lower bound for the two-component charged
Bose gas. Let us recall this result. Given a vector
$e=(e_1,\ldots,e_N)\in\{-1,1\}^N$ (representing charges) we introduce the
Hamiltonian
$$
H^{(N)}(e) = \sum_{j=1}^N (-\Delta_j) + \sum_{i<j} \frac{e_i e_j}{|x_i-x_j|}
\qquad\text{in}\ L^2(\R^{3N}) \,.
$$
In \cite{LiSo2} it is proved that
\begin{equation}
\label{eq:liebsolo}
\inf_{e\in\{-1,1\}^N} \inf\spec H^{(N)}(e) \geq -A N^{7/5} \left( 1+ o(1)\right) \,,
\end{equation}
where $A$ is as in \eqref{eq:a}. (Note that we rescaled the result from \cite{LiSo2}, where the kinetic energy is described by $-\Delta/2$, whereas it is $-\Delta$ in our case.)

Let us fix $N\in\N$. Given $\psi\in H^1(\R^{3N})$ with
$\int_{\R^{3N}}|\psi|^2\,dx = 1$ we define $\tilde\psi(x,y) =
\psi(x)\psi(y)$ for $x,y\in\R^{3N}$ and let $e=(1,\ldots,1,-1,\ldots,-1)$,
where both $1$ and $-1$ are repeated $N$ times. Then, from
\eqref{eq:energy},
$$
\left( \tilde\psi, H^{(2N)}(e) \tilde\psi \right) = 2\, \mathcal E^{(N)}_1[\psi]
\qquad\text{and}\qquad
\int_{\R^{6N}} |\tilde\psi|^2 \,dx\,dy = 1 \,.
$$
Thus, we conclude that
$$
2 E_1^{(b)}(N)=  \inf_{\psi} 2\, \mathcal E^{(N)}_1[\psi]  \geq
\inf_{e\in\{-1,1\}^N} \inf\spec H^{(2N)}(e) \,.
$$
The lower bound \eqref{eq:main} in Theorem \ref{main} thus follows from
\eqref{eq:liebsolo}.


\subsection{The upper bound}

\subsubsection{Introduction}

Most of this paper is taken up with the upper bound in Theorem~\ref{main}.
Normally, upper bounds are easier than lower bounds, but this is not
necessarily so
for Coulomb systems where we want not just an asymptotic power law but
also an accurate constant multiplying the power law. The truly remarkable
fact is that the accurate constants were first found by Foldy
\cite{Fo} for the one-component plasma (jellium) using Bogolubov's method,
and for the two-component gas by Dyson \cite{D} using Foldy's result. None
of this was rigorous, however. The rigorous lower bounds were done in 
\cite{LiSo1, LiSo2}. The upper bounds were done by Solovej in a tour de force \cite{So}. Our work here consists largely in imitating
and adapting Solovej's work to our special case.

As Solovej points out, Foldy's calculation, while not yielding a rigorous
lower bound, essentially yields a rigorous upper bound -- much as
Pekar's
model is a rigorous  upper bound for Fr\"ohlich's model \eqref{froh}.
Unfortunately, this is not quite so simple  since one of the things done in
\cite{Fo} is to use periodic boundary conditions (which is not easy to
justify for Coulomb systems). Another hard to justify procedure is to mimic
the charge neutralizing
background by simply discarding the $k=0$ term in the Fourier series for 
the Coulomb potential. It was Solovej who succeeded in solving these
problems.

Bogolubov's method of 1947 \cite{Bo} takes account of bosonic symmetry using
boson
creation and annihilation operators in momentum space. The Coulomb
interaction is, like any 2-body interaction, represented by a quartic in
these operators. Bogolubov's idea is to retain only those terms that have
no more than two operators of non-zero momentum and to replace the zero
momentum operators by $\sqrt{N}$ (see \cite{LSY}). The resulting quadratic
in non-zero momentum operators is then diagonalized. This latter process can
be thought of as using  `squeezed coherent states'.

\subsubsection{First Step}
In our proof of the upper bound, we shall linearize the non-linear
functional $\mathcal E^{(N)}_1[\cdot]$, as in \cite{FLST}.
This process replaces the 2-body interaction by a one-body potential, so
that the problem tends to resemble the one-component Coulomb gas.
Given a real-valued function $\sigma$ on $\R^3$
with $D(\sigma,\sigma)<\infty$ we introduce the operator
$$
H^{(N)}_\sigma = \sum_{j=1}^N \left(-\Delta_j - \sigma*|x_j|^{-1} \right) + \sum_{i<j} |x_i-x_j|^{-1} + D(\sigma,\sigma)
\qquad\text{in}\ L^2_{\mathrm{symm}}(\R^{3N}) \,,
$$
where $\sigma*|x_j|^{-1}$ is an abbreviation for $(\sigma*|\cdot|^{-1})(x_j)$. Then by linearization we mean that
\begin{equation}
\label{eq:linear}
\inf_\sigma \inf\spec H^{(N)}_\sigma = E^{(b)}_1(N) \,,
\end{equation}
as is easily verified by completing a square.

We denote by $\mathcal F(L^2(\R^3))$ the bosonic Fock space over $L^2(\R^3)$ and consider the operator
$$
\mathcal H_\sigma = \bigoplus_{N=0}^\infty H^{(N)}_\sigma
\qquad\text{in}\ \mathcal F(L^2(\R^3))
$$
(with $H^{(0)}_\sigma=0$ and $H^{(1)}_\sigma=-\Delta-\sigma*|x|^{-1} + D(\sigma,\sigma)$.) As usual, we denote by $\mathcal N$ the number operator.

In order to prove the upper bound in Theorem \ref{main}, our strategy will be to find an upper bound on $\inf_\sigma \inf\spec\mathcal H_\sigma$ of the required form. Indeed, the main ingredient in the proof of Theorem \ref{main} is the following proposition.

\begin{proposition}\label{grandcan}
For any sufficiently large $n$ there is a normalized $\Psi_{n}\in\mathcal F(L^2(\R^3))$ with finite kinetic energy and a $\sigma_{n}$ with $D(\sigma_{n},\sigma_{n})<\infty$ such that
\begin{align}
\label{eq:grandcan1}
\left( \Psi_{n},\mathcal H_{\sigma_{n}} \Psi_{n} \right) & \leq -A n^{7/5} \left(1- C n^{-1/35} \right) \,,\\
\label{eq:grandcan2}
\left(\Psi_{n},\mathcal N\Psi_{n}\right) -n &  \leq C n^{3/5} \,, \\
\label{eq:grandcan3}
\left(\Psi_{n},\mathcal N^2 \Psi_{n}\right)- 
\left(\Psi_{n},\mathcal N\Psi_{n}\right)^2 & \leq C n\,,
\end{align}
where $A$ is the constant from \eqref{eq:a} and $C$ is some constant independent of $n$.
\end{proposition}

Actually, instead of \eqref{eq:grandcan2} and \eqref{eq:grandcan3} we shall show the stronger facts that
\begin{equation}
\label{eq:grandcan2b}
\left|\left(\Psi_{n},\mathcal N\Psi_{n}\right) -n \right| \leq C n^{3/5}
\end{equation}
\begin{equation}
\label{eq:grandcan3b}
\left(\Psi_{n},\mathcal N^2 \Psi_{n}\right)- 
\left(\Psi_{n},\mathcal N\Psi_{n}\right)^2 \leq n + C n^{3/5+4/15} \,.
\end{equation}

\subsubsection{From Proposition \ref{grandcan} to Theorem \ref{main}. }
Accepting Proposition \ref{grandcan} for the moment, we now explain how this trial state $\Psi_{n}$ on $\mathcal F$ leads to an upper bound for $E^{(b)}_1(N)$. The following argument is taken from \cite{So} and reproduced here for the sake of completeness.

\begin{proof}[Proof of Theorem \ref{main} given Proposition \ref{grandcan}]
As a preliminary to the proof, we shall show that for all $M>0$ and all sufficiently large $n$,
\begin{equation}
\label{eq:reductionsum}
\sum_{m>\left(\Psi_{n},\mathcal N\Psi_{n}\right)+M} m^{7/5} \left\|\Psi_{n}^{(m)} \right\|^2 \leq C M^{-3/5} n^{17/10} \,.
\end{equation}
Here $\Psi_{n}^{(m)}$ denotes the projection of $\Psi_{n}$ onto the sector of $m$ particles.

Indeed, by H\"older's inequality,
\begin{align*}
& \sum_{m>\left(\Psi_{n},\mathcal N\Psi_{n}\right)+M} m^{7/5} \left\|\Psi_{n}^{(m)} \right\|^2 \leq M^{-3/5} \sum_{m=0}^\infty m^{7/5} |m-\left(\Psi_{n},\mathcal N\Psi_{n}\right) |^{3/5} \left\|\Psi_{n}^{(m)} \right\|^2 \\
& \qquad \qquad\qquad\leq M^{-3/5} \left(\Psi_{n},\mathcal N^2 \Psi_{n}\right)^{7/10} \left( \left(\Psi_{n},\mathcal N^2 \Psi_{n}\right)- 
\left(\Psi_{n},\mathcal N\Psi_{n}\right)^2 \right)^{3/10} \,.
\end{align*}
It follows from \eqref{eq:grandcan2} and \eqref{eq:grandcan3} that
$$
\left(\Psi_{n}, \mathcal N^2 \Psi_{n}\right) \leq C n^2
\quad\text{and}\quad
\left(\Psi_{n},\mathcal N^2 \Psi_{n}\right)- 
\left(\Psi_{n},\mathcal N\Psi_{n}\right)^2 \leq C n \,.
$$
This proves \eqref{eq:reductionsum}. \endproof

{\it We now return to the proof of Theorem \ref{main}.} Recalling the
linearization formula \eqref{eq:linear} and using the fact that $N\mapsto
\inf\spec H^{(N)}_\sigma$ is non-increasing and non-positive, we have  for
any $n$,  
\begin{align*}
E^{(b)}_1(N) & \leq \inf\spec H^{(N)}_{\sigma_{n}} \leq \left( \inf\spec H^{(N)}_{\sigma_{n}} \right) \sum_{m=0}^N \left\|\Psi_{n}^{(m)}\right\|^2 
\leq \sum_{m=0}^N \left( \inf\spec H^{(m)}_{\sigma_{n}}\right) \left\|\Psi_{n}^{(m)}\right\|^2 \\
& \leq \sum_{m=0}^N \left( \Psi_{n}^{(m)}, H_{\sigma_{n}}^{(m)} \Psi_{n}^{(m)} \right) = \left( \Psi_{n},\mathcal H_{\sigma_{n}} \Psi_{n} \right) - \sum_{m=N+1}^\infty \left( \Psi_{n}^{(m)}, H_{\sigma_{n}}^{(m)} \Psi_{n}^{(m)} \right) \,.
\end{align*}
Thus, we need a lower bound on the last sum. Because of the linearization formula \eqref{eq:linear} and the lower bound on $E_1^{(b)}(m)$ proved above we have
$$
\sum_{m=N+1}^\infty \left( \Psi_{n}^{(m)}, H_{\sigma_{n}}^{(m)} \Psi_{n}^{(m)} \right)
\geq \sum_{m=N+1}^\infty \|\Psi_{n}^{(m)}\|^2\ \mathcal E_1^{(m)}[ \Psi_{n}^{(m)} ] 
\geq -C \sum_{m=N+1}^\infty m^{7/5} \left\|\Psi_{n}^{(m)}\right\|^2  \,.
$$
For all $N\in\N$ sufficiently large we shall apply the previous bounds with $n=N-C_0 N^{3/5}$. Here, by \eqref{eq:grandcan2}, the constant $C_0$ can be chosen in such a way that for another constant $C_1>0$ one has
$$
\left(\Psi_{n},\mathcal N\Psi_{n}\right) \leq N-C_1 N^{3/5} \,.
$$
Then we can use \eqref{eq:reductionsum} to bound
\begin{align*}
-C \sum_{m=N+1}^\infty m^{7/5} \left\|\Psi_{n}^{(m)}\right\|^2 & \geq -C \sum_{m>\left(\Psi_{n},\mathcal N\Psi_{n}\right)+C_1 N^{3/5}} m^{7/5} \left\|\Psi_{n}^{(m)}\right\|^2 \\
& \geq -C N^{-9/25} n^{17/10} \\
& \geq -C N^{7/5-3/50} \,.
\end{align*}
Using \eqref{eq:grandcan1} we conclude that
$$
N^{-7/5} E_1^{(b)}(N) \leq - A \left(1-C_0 N^{-2/5}\right)^{7/5} \left( 1- C N^{-1/35} \right) + C N^{-3/50}  \leq A \left( 1- C N^{-1/35}\right) \,.
$$
This proves \eqref{eq:mainupper} and completes the proof of Theorem \ref{main}.
\end{proof}

\subsubsection{Proof of Proposition \ref{grandcan}, Step 1}

Thus, we have reduced the proof of the upper bound in Theorem \ref{main} to
the proof of Proposition \ref{grandcan}. The following lemma guarantees the
existence of an appropriate trial state.

\begin{lemma}\label{trial}
For any real $\phi\in H^1(\R^3)$ and any non-negative, real trace class operator $\gamma$ with finite kinetic energy there is a normalized $\Psi\in\mathcal F(L^2(\R^3))$ with finite kinetic energy and a $\sigma$ with $D(\sigma,\sigma)<\infty$ such that
\begin{align}
\label{eq:trial1}
\left( \Psi,\mathcal H_\sigma \Psi\right) & = (\phi,-\Delta\phi) + \tr(-\Delta)\gamma - \tr K_\phi\left(\sqrt{\gamma(\gamma+1)}-\gamma \right) \notag \\
& \qquad + \frac12 \iint_{\R^3\times\R^3} \frac{|\gamma(x,x')|^2}{|x-x'|} \,dx\,dx' \notag \\
& \qquad + \frac12 \iint_{\R^3\times\R^3} \frac{|\sqrt{\gamma(\gamma+1)}(x,x')|^2}{|x-x'|} \,dx\,dx' \,, \\
\label{eq:trial2}
\left(\Psi,\mathcal N\Psi\right) & = \|\phi\|^2 + \tr\gamma \,,\\
\label{eq:trial3}
\left(\Psi,\mathcal N^2\Psi\right) -\left(\Psi,\mathcal N\Psi\right)^2 & = \|\phi\|^2 + 2\tr\gamma(\gamma+1) - 2 (\phi,\sqrt{\gamma(\gamma+1)}\phi) + 2(\phi,\gamma\phi) \,.
\end{align}
Here, $K_\phi$ is the integral operator with integral kernel $K_\phi(x,x')=\phi(x)|x-x'|^{-1}\phi(x')$.
\end{lemma}

\begin{proof}
We write $\gamma= \sum_{\alpha=1}^\infty \frac{\lambda_\alpha^2}{1-\lambda_\alpha^2}|\psi_\alpha\rangle\langle\psi_\alpha|$ with $0<\lambda_\alpha<1$ and $(\psi_\alpha)$ orthonormal. Since $\gamma$ is real, the $\psi_\alpha$ can be chosen real. Following \cite{So} we set
$$
\Psi = \prod_\alpha \left( (1-\lambda_\alpha^2)^{1/4} \exp\left( -\frac{\lambda_\alpha}{2}\left(a\left(\psi_\alpha\right)^*-\left(\phi,\psi_\alpha\right)\right)\left( a\left(\psi_\alpha\right)^* -\left(\phi,\psi_\alpha\right)\right)\right)\right) |\phi\rangle_C
$$ 
with
$$
|\phi\rangle_C = \exp\left(-\frac12\|\phi\|^2 + a\left(\phi\right)^*\right)|0\rangle \,.
$$
Here $a$ and $a^*$ are (bosonic) annihilation and creation operators on $\mathcal F(L^2(\R^3)$. One can check that $\|\Psi\|=1$. Equations \eqref{eq:trial2} and \eqref{eq:trial3} follow from \cite[(23), (24)]{So}. (Note that the last two terms on the right side of \eqref{eq:trial3} are absent in \cite{So}, since there $\gamma\phi=0$.) Moreover, as in \cite[(58), (59) and (60)]{So}, we find
\begin{align*}
\left(\Psi, \bigoplus_{N=0}^\infty \sum_{j=1}^N\left(-\Delta_j -\sigma*|x_j|^{-1}\right)\Psi \right) & = \|\nabla\phi\|^2 + \tr(-\Delta-\sigma*|x|^{-1})\gamma \\
& = \|\nabla\phi\|^2 + \tr(-\Delta)\gamma - 2 D(\sigma,\rho_\gamma)
\end{align*}
and that
\begin{align*}
& \left( \Psi,\bigoplus_{N=0}^\infty \sum_{i<j} |x_i-x_j|^{-1}\Psi\right) = - \tr K_\phi\left(\sqrt{\gamma(\gamma+1)}-\gamma\right) \\
& \qquad\qquad + D(\phi^2,\phi^2) + 2D(\rho_\gamma,\phi^2) + D(\rho_\gamma,\rho_\gamma) \\
& \qquad\qquad + \frac12 \iint_{\R^3\times\R^3} \frac{|\gamma(x,x')|^2}{|x-x'|} \,dx\,dx' + \frac12 \iint_{\R^3\times\R^3} \frac{|\sqrt{\gamma(\gamma+1)}(x,x')|^2}{|x-x'|} \,dx\,dx' \,.
\end{align*}
With the choice $\sigma= \phi^2 + \rho_\gamma$ we obtain \eqref{eq:trial1}.
\end{proof}

\subsubsection{Proof of Proposition \ref{grandcan}, Step 2.}
This reduces our task of proving Proposition \ref{grandcan} to finding corresponding $\phi$ and $\gamma$. We will do this using the method of coherent states; see, e.g., \cite[Sec. 12]{LL}. Given a real, even function $G\in H^1(\R^3)$ with $\|G\|=1$ we let
$$
G_{p,q}(x) = e^{ip\cdot x} G(x-q) \,,\qquad p,q,x\in\R^3 \,.
$$
Let $M$ be a non-negative, integrable function on $\R^3\times\R^3$ satisfying $M(p,q)=M(-p,q)$ and define the operator
$$
\gamma = \iint_{\R^3\times\R^3} M(p,q) |G_{p,q}\rangle\langle G_{p,q}| \,\frac{dp\,dq}{(2\pi)^3} \,.
$$
Clearly, $\gamma$ is a real, non-negative trace class operator. Let $\phi\in H^1(\R^3)$ be real. Then Lemma \ref{trial} yields a trial state $\Psi$ and a $\sigma$ with
$$
\left(\Psi,\mathcal H_\sigma\Psi\right) = (\phi,-\Delta\phi) + \tr(-\Delta)\gamma - \tr K_\phi\left(\sqrt{\gamma(\gamma+1)}-\gamma\right)  + \mathcal R_{\mathrm{xc}} \,,
$$
where
$$
\mathcal R_{\mathrm{xc}} = \frac12 \iint_{\R^3\times\R^3} \frac{|\gamma(x,x')|^2}{|x-x'|} \,dx\,dx' + \frac12 \iint_{\R^3\times\R^3} \frac{|\sqrt{\gamma(\gamma+1)}(x,x')|^2}{|x-x'|} \,dx\,dx' \,.
$$
By \cite[Thm. 12.9]{LL}
$$
\tr(-\Delta)\gamma = \tr(-\Delta-\|\nabla G\|^2)\gamma + \mathcal R_\loc = \iint_{\R^3\times\R^3} p^2 M(p,q)\,\frac{dp\,dq}{(2\pi)^3} + \mathcal R_\loc \,,
$$
where
$$
\mathcal R_\loc = \|\nabla G\|^2 \tr\gamma = \|\nabla G\|^2 \iint_{\R^3\times\R^3} M(p,q) \,\frac{dp\,dq}{(2\pi)^3} \,.
$$
Moreover, since $t\mapsto\sqrt{t(t+1)}-t$ is operator-concave, Solovej's operator version of the Berezin--Lieb inequality \cite[Thm. A.1]{So} yields
\begin{align*}
& \tr K_\phi\left(\sqrt{\gamma(\gamma+1)}-\gamma\right) \\
& \qquad\qquad \geq \iint_{\R^3\times\R^3} \left( \sqrt{M(p,q)(M(p,q)+1)}-M(p,q) \right) \left( G_{p,q},K_\phi G_{p,q}\right) \,\frac{dp\,dq}{(2\pi)^3} \,.
\end{align*}
It is not unreasonable to think that $\left( G_{p,q},K_\phi G_{p,q}\right)$ should be an approximation to $4\pi \phi(q)^2 |p|^{-2}$, and therefore we introduce the remainder
$$
\mathcal R_{\mathrm{int}} \!=\! \iint_{\R^3\times\R^3} \!\!\left( \sqrt{M(p,q)(M(p,q)+1)}-M(p,q) \right)\! \left( \frac{4\pi \phi(q)^2}{|p|^2} - \left( G_{p,q},K_\phi G_{p,q}\right) \right)\! \frac{dp\,dq}{(2\pi)^3} \,.
$$
Thus, we have
$$
\left(\Psi,\mathcal H_\sigma\Psi\right) = \mathcal E(\phi,M) + \mathcal R_\loc + \mathcal R_{\mathrm{int}} + \mathcal R_{\mathrm{xc}}
$$
where
\begin{align*}
\mathcal E(\phi,M) = & \|\nabla\phi\|^2 + \iint_{\R^3\times\R^3} p^2 M(p,q)\,\frac{dp\,dq}{(2\pi)^3} \\
& - 4\pi \iint_{\R^3\times\R^3} \left( \sqrt{M(p,q)(M(p,q)+1)}-M(p,q) \right) \frac{\phi(q)^2}{p^2} \,\frac{dp\,dq}{(2\pi)^3} \,.
\end{align*}

\medskip

\emph{Minimizing $\mathcal E(\phi,M)$.}
In order to make our upper bound as small as possible, we would like to minimize the functional $\mathcal E(\phi,M)$ with respect to functions $M\geq 0$ and $\phi\in H^1(\R^3)$ satisfying the requirements above. Carrying out the minimization over $M$ yields
\begin{equation}
\label{eq:choicem}
M_*(p,q) = g\left( \frac{|p|}{(4\pi)^{1/4} \phi(q)^{1/2}} \right) \,,
\qquad\text{where}\qquad
g(a) = \frac12 \left(\frac{a^4+1}{\sqrt{a^4(a^4+2)}} - 1\right) \,.
\end{equation}
With this choice of $M$ we obtain
\begin{align*}
& \iint_{\R^3\times\R^3} \left( p^2 M_*(p,q) - \frac{4\pi \phi(q)^2}{p^2} \left( \sqrt{M_*(p,q)(M_*(p,q)+1)}-M_*(p,q) \right)\right) \frac{dp\,dq}{(2\pi)^3} \\
& \qquad = \left( \frac{4}{\pi} \right)^{3/4} \int_{\R^3} |\phi(q)|^{5/2} \,dq\ 
\int_0^\infty \left( a^4 g - \left( \sqrt{g(g+1)} - g\right)\right) da \\ 
& \qquad = - 2^{1/2} \pi^{-3/4} \int_{\R^3} |\phi(q)|^{5/2} \,dq\ 
\int_0^\infty \left( a^4 + 1 - a^2 \sqrt{a^4+2} \right) da \\ 
& \qquad = - \frac25 \left( \frac{2}{\pi} \right)^{1/4} \frac{\Gamma(3/4)}{\Gamma(5/4)} \ \int_{\R^3} |\phi(q)|^{5/2} \,dq \\
& \qquad = - I_0 \ \int_{\R^3} |\phi(q)|^{5/2} \,dq \,,
\end{align*}
with $I_0$ from \eqref{eq:i0}. Thus,
\begin{equation}
\label{eq:min}
\mathcal E(\phi,M_*)=\|\nabla\phi\|^2 - I_0 \int_{\R^3} |\phi(q)|^{5/2} \,dq \,.
\end{equation}
The latter functional has a minimizer for any fixed value of $\|\phi\|^2$ and the minimizer is non-negative. (This is a well-known result in the calculus of variations -- in fact, for us the existence of a minimizer is not really necessary and we could simply work with almost-minimizers.) Thus, let us introduce a parameter $n>0$ and let us choose $\phi_*$ to be the minimizer of \eqref{eq:min} under the constraint $\|\phi\|^2 = n$. Then, by scaling,
$$
\phi_*(x) = n^{4/5} \Phi(n^{1/5} x)
$$
for a universal function $\Phi$ with $\|\Phi\|=1$, and \eqref{eq:min} is equal to $-A n^{7/5}$ with $A$ from \eqref{eq:a}. 

Moreover, if $M_*$ is chosen according to \eqref{eq:choicem}, then
$$
\tr\gamma = \iint_{\R^3\times\R^3} M_*(p,q) \,\frac{dp\,dq}{(2\pi)^3} = \frac{(4\pi)^{3/4}}{2\pi^2} \int_{\R^3} \phi_*(q)^{3/2}\,dq \int_0^\infty g(a) a^2 \,da = C n^{3/5} \,.
$$
(The fact that $\Phi\in L^{3/2}$ follows from the fact that $\Phi$ is exponentially decaying, as can be verified using the Euler--Lagrange equation satisfied by $\Phi$.) Thus, from \eqref{eq:trial1} we obtain
$$
\left( \Psi,\mathcal N\Psi\right) = n + C n^{3/5} \,.
$$

\medskip

\emph{Definition of $\gamma_{n,\epsilon}$.}
The problem with the above argument is that we cannot get a good bound on $\tr\gamma^2$, which is needed in order to control the fluctuations of the particle number of $\Psi$, see \eqref{eq:trial3}. Therefore, we shall introduce $g_\epsilon(a)=0$ if $a\leq\epsilon$ and $g_\epsilon(a)=g(a)$ if $a>\epsilon$. We denote by $\Psi_{n,\epsilon}$ the state constructed in Lemma \ref{trial} corresponding to $\gamma_{n,\epsilon}$ which is given in terms of
$$
M_{n,\epsilon}(p,q) = g_\epsilon\left(\frac{|p|}{(4\pi)^{1/4}n^{2/5} \Phi(n^{1/5}q)^{1/2}}\right) \,.
$$
Then, as before, $\gamma_{n,\epsilon}$ is trace class with
\begin{equation}
\label{eq:gammatrace}
\tr\gamma_{n,\epsilon} = \frac{(4\pi)^{5/4}}{2\pi^2} \int_{\R^3} \phi(q)^{3/2}\,dq \int_\epsilon^\infty g(a) a^2 \,da \leq C n^{3/5} \,,
\end{equation}
with $C$ independent of $\epsilon$. In view of \eqref{eq:trial2} this implies \eqref{eq:grandcan2b}, which in turn implies \eqref{eq:grandcan2}.

The advantage of introducing the parameter $\epsilon>0$ is that now, by the Berezin--Lieb inequality \cite{B,L},
\begin{equation}
\label{eq:gammasquared}
\tr \gamma_{n,\epsilon}^2  \leq \iint_{\R^3\times\R^3} M(p,q)^2 \,\frac{dp\,dq}{(2\pi)^3}
= \frac{(4\pi)^{3/4}}{2\pi^2} \int_{\R^3} \phi(q)^{3/2}\,dq \int_\epsilon^\infty g(a)^2 a^2 \,da \leq C \epsilon^{-1} n^{3/5} \,.
\end{equation}
(In the final inequality we used the fact that $g(a)$ diverges like $a^{-2}$ as $a\to 0$.) Thus,
$$
0\leq 2 \tr\gamma_{n,\epsilon}(\gamma_{n,\epsilon}+1) \leq C \epsilon^{-1} n^{3/5} \,.
$$
We will later choose $\epsilon= n^{-4/15}$. Then, in view of \eqref{eq:trial3}, and since $(\phi,\sqrt{\gamma(\gamma+1)}\phi) \geq (\phi,\gamma\phi)$, this implies \eqref{eq:grandcan3b}, which in turn implies \eqref{eq:grandcan3}.

Thus, to complete the proof of Proposition \ref{grandcan}, we need to verify that, if $\epsilon$ is chosen suitably as function of $n$, then
\begin{equation}
\label{eq:upperboggoal}
\left(\Psi_{n,\epsilon},\mathcal H\Psi_{n,\epsilon}\right) \leq -A n^{7/5} (1 - C n^{-1/35})
\qquad\text{as}\ n\to\infty \,.
\end{equation}
Note that by repeating the previous argument we find that
\begin{align*}
\mathcal E(\phi, M_{n,\epsilon}) & = \|\nabla\phi\|^2 - I_\epsilon \|\phi\|_{5/2}^{5/2} \\
& = \left( \|\nabla\phi\|^2 - I_0 \|\phi\|_{5/2}^{5/2} \right) + (I_0 - I_\epsilon) \|\phi\|_{5/2}^{5/2} \\
& = -A n^{7/5} + \mathcal R_{\mathrm{main}} \,,
\end{align*}
where
\begin{align*}
- I_\epsilon & = \left( \frac{4}{\pi} \right)^{3/4} \ \int_\epsilon^\infty \left( a^4 g - \left( \sqrt{g(g+1)} - g\right)\right) da \\
& = - 2^{1/2} \pi^{-3/4} \int_{\R^3} |\phi(q)|^{5/2} \,dq\ 
\int_\epsilon^\infty \left( a^4 + 1 - a^2 \sqrt{a^4+2} \right) da
\end{align*}
and
$$
\mathcal R_{\mathrm{main}} = (I_0 - I_\epsilon) \|\phi\|_{5/2}^{5/2} \,.
$$
This will prove \eqref{eq:upperboggoal}, provided we can show that, for an appropriate choice of the function $G$, the errors $\mathcal R_{\mathrm{main}}$, $\mathcal R_\loc$, $\mathcal R_{\mathrm{int}}$ and $\mathcal R_{\mathrm{xc}}$ are at most $O(n^{7/5-1/35})$. To do so, we choose $G(x) = (\pi\ell)^{-3/2} \exp(-(x/\ell)^2)$ with a parameter $\ell>0$ to be determined.

\medskip

\emph{Bound on $\mathcal R_{\mathrm{main}}$.}
Since $a^4 + 1 + a^2 \sqrt{a^4+2}$ is finite near $a=0$ and since, by scaling, $\|\phi\|_{5/2}^{5/2} = n^{7/5}\|\Phi\|_{5/2}^{5/2}$, we have
$$
\mathcal R_{\mathrm{main}} \leq C \epsilon n^{7/5} \,.
$$

\medskip

\emph{Bound on $\mathcal R_\loc$.}
Clearly, by \eqref{eq:gammatrace} we have
$$
\mathcal R_\loc = \|\nabla G\|^2 \iint_{\R^3\times\R^3} M_{n,\epsilon}(p,q)\,\frac{dp\,dq}{(2\pi)^3} \leq C \ell^{-2} n^{3/5} = C n^{7/5} (n^{2/5} \ell)^{-2} \,.
$$

\medskip

\emph{Bound on $\mathcal R_{\mathrm{int}}$.} This bound can be taken literally from \cite[(47)]{So},
$$
\mathcal R_{\mathrm{int}} \leq C n^{7/5} \left( \left( n^{2/5} \ell\right)^{-1/2} + \left(n^{2/5} \ell\right)^3 n^{-1/5} \right) \,.
$$
(The constant here can be chosen independently of $\epsilon\in (0,1]$.)

\medskip

\emph{Bound on $\mathcal R_{\mathrm{xc}}$.} Here we argue as in Solovej's analysis of the one-component gas; see \cite[(67)]{So}. Hardy's inequality yields
\begin{align*}
\iint_{\R^3\times\R^3} \frac{|\gamma_{n,\epsilon}(x,x')|^2}{|x-x'|} \,dx\,dx'
& \leq \left( \iint_{\R^3\times\R^3} |\gamma_{n,\epsilon}(x,x')|^2 \,dx\,dx'
\right)^{1/2} \\
& \qquad\qquad\times
\left( \iint_{\R^3\times\R^3} \frac{|\gamma_{n,\epsilon}(x,x')|^2}{|x-x'|^2} \,dx\,dx' \right)^{1/2} \\
& \leq 2 \left( \tr \gamma_{n,\epsilon}^2\right)^{1/2} \left( \tr(-\Delta)\gamma_{n,\epsilon}^2 \right)^{1/2} \,.
\end{align*}
According to \eqref{eq:gammasquared} we have $\tr\gamma_{n,\epsilon}^2 \leq C \epsilon^{-1} n^{3/5}$. Moreover, by Solovej's operator-version of the Berezin--Lieb inequality \cite[Thm. A.1]{So}, we have
\begin{align*}
\tr(-\Delta)\gamma_{n,\epsilon}^2
& \leq \iint_{\R^3\times\R^3} M_{n,\epsilon}(p,q)^2 (G_{p,q},(-\Delta) G_{p,q}) \,\frac{dp\,dq}{(2\pi)^3} \\
& = \iint_{\R^3\times\R^3} M_{n,\epsilon}(p,q)^2 \left( p^2 + \|\nabla G\|^2 \right) \,\frac{dp\,dq}{(2\pi)^3} \\
& \leq C \left( n^{7/5} + \epsilon^{-1} \ell^{-2} n^{3/5} \right) \,.
\end{align*}
(Here we used the fact that $\int_0^\infty a^4 g(a)^2 \,da<\infty$ to bound the term involving $p^2$, as well as \eqref{eq:gammasquared} to bound the term involving $\|\nabla G\|^2$.) Thus,
$$
\iint_{\R^3\times\R^3} \frac{|\gamma_{n,\epsilon}(x,x')|^2}{|x-x'|} \,dx\,dx' \leq C \epsilon^{-1/2} n^{7/5-2/5} \left( 1+ \epsilon^{-1/2} \left(n^{2/5}\ell\right)^{-1} \right) \,.
$$
The term that involves $\sqrt{\gamma_{n,\epsilon}(\gamma_{n,\epsilon}+1)}$ instead of $\gamma_{n,\epsilon}$ can be bounded similarly and we finally obtain
$$
\mathcal R_{\mathrm{xc}} \leq C \epsilon^{-1/2} n^{7/5-2/5} \left( 1+ \epsilon^{-1/2} \left(n^{2/5}\ell\right)^{-1} \right) \,.
$$

In order to minimize the remainder in $\mathcal R_{\mathrm{int}}$ we choose
$\ell = n^{-2/5+2/35}$ and find $\mathcal R_{\mathrm{int}} \leq C
n^{7/5-1/35}$ and $\mathcal R_\loc \leq C n^{7/5-4/35}$. In order to
minimize the error coming from $\mathcal R_{\mathrm{main}}$ and $\mathcal
R_{\mathrm{xc}}$ we choose $\epsilon= n^{-4/15}$ and find $\mathcal
R_{\mathrm{main}} \leq C n^{7/5-4/15}$ and $\mathcal R_{\mathrm{xc}} \leq C
n^{7/5-4/15}$. As explained above, this proves \eqref{eq:upperboggoal} and
finishes the proof of Theorem \ref{main}.\qquad\qquad \qedsymbol


\appendix

\section{Proof of Theorem \ref{main2}}

Since \cite{GM} have already shown the upper bound, we only need to show the lower bound. In fact, we shall show the lower bound
\begin{equation}
\label{eq:main2proofgoal}
\mathcal E^{(N)}_{U,\alpha}[\psi] \geq N^{7/3} e^{(f)}_U - C (1-U)^2 N^{7/3-2/33} \left( 1 + N^{-40/33} (U/(1-U))^2 \right) \,.
\end{equation}

 Using the Lieb--Oxford inequality \cite{LO} we bound from below
$$
\mathcal E^{(N)}_{U}[\psi] \geq \tr(-\Delta)\gamma_\psi - (1-U) D(\rho_\gamma,\rho_\gamma) - 1.68 \, U \,  \int_{\R^3} \rho_\gamma^{4/3} \,dx \,,
$$
where
$$
\gamma_\psi(x,x')= \int\cdots\int_{\R^{3(N-1)}} \overline{\psi(x,x_2,\ldots,x_N)}\psi(x',x_2,\ldots,x_N)\,dx_2\ldots dx_N
$$
denotes the one-particle density matrix. Thus, for any $G\in H^1(\R^3)$ and any $0<\epsilon<1$,
$$
\mathcal E^{(N)}_{U} \geq (1-\epsilon) \tr(-\Delta+\|\nabla G\|^2)\gamma_\psi -(1-U) D(\rho_\psi*|G|^2,\rho_\psi*|G|^2) +\epsilon\mathcal T - \mathcal R_\loc - \mathcal R_{\mathrm{rep}} - \mathcal R_{\mathrm{xc}} \,,
$$
where
\begin{align*}
\mathcal T & = \tr(-\Delta)\gamma_\psi \,,\\
\mathcal R_\loc^{(\epsilon)} & = (1-\epsilon) \|\nabla G\|^2 \tr\gamma_\psi = (1-\epsilon) N\|\nabla G\|^2 \,, \\
\mathcal R_{\mathrm{rep}} & = -(1-U) \left( D(\rho_\psi*|G|^2,\rho_\psi*|G|^2) - D(\rho_\psi,\rho_\psi) \right) \,,\\ 
\mathcal R_{\mathrm xc}^{(\epsilon)} & = 1.68\, U \int_{\R^3} \rho_\gamma^{4/3} \,dx \,.
\end{align*}
Now assume again the $G$ is real, even and normalized and let $G_{p,q}$ be the corresponding coherent states. Set
$$
M(p,q) = (G_{p,q},\gamma_\psi G_{p,q}) \,.
$$
Then $0\leq\gamma_\psi\leq 1$ implies that $0\leq M\leq 1$. We now observe that for any number $\rho>0$,
\begin{equation*}
\label{eq:tfcomputation}
\inf\left\{ \int_{\R^3} p^2 m(p) \,\frac{dp}{(2\pi)^3} :\ 0\leq m\leq 1\,,\ \int_{\R^3} m(p) \,\frac{dp}{(2\pi)^3} = \rho \right\} = \frac35 (6\pi^2)^{2/3} \rho^{5/3} \,.
\end{equation*}
(In fact, the infimum is attained iff $m(p) =\chi_{\{ p^2<(6\pi^2\rho)^{2/3}\}}$.) Since $\int_{\R^3} M(p,q)\,\frac{dp}{(2\pi)^3} = (\rho_\psi*G^2)(q)$ for any $q\in\R^3$, we obtain the lower bound
\begin{align*}
\tr(-\Delta+\|\nabla g\|^2)\gamma_\psi = \iint_{\R^3\times\R^3} p^2 M(p,q)\,\frac{dp\,dq}{(2\pi)^3} \geq \frac{3}{5} (6\pi^2)^{2/3} \int_{\R^3} (\rho_\gamma* G^2)^{5/3} \,dx \,.
\end{align*}
Thus,
\begin{align*}
\mathcal E^{(N)}_{U}[\psi] & \geq (1-\epsilon) \frac{3}{5} (6\pi^2)^{2/3} \int_{\R^3} (\rho_\gamma* G^2)^{5/3} \,dx  - (1-U)D(\rho_\gamma*G^2,\rho_\gamma*G^2) \\
& \qquad +\epsilon\mathcal T - \mathcal R_\loc - \mathcal R_{\mathrm{rep}} - \mathcal R_{\mathrm{xc}} \\
& \geq (1-\epsilon)^{-1} N^{7/3} e^{(f)}_U +\epsilon\mathcal T - \mathcal R_\loc - \mathcal R_{\mathrm{rep}} - \mathcal R_{\mathrm{xc}} \\
& \geq N^{7/3} e^{(f)}_U +\epsilon\mathcal T -\mathcal R_{\mathrm{main}} - \mathcal R_\loc - \mathcal R_{\mathrm{rep}} - \mathcal R_{\mathrm{xc}} \,.
\end{align*}
with
$$
\mathcal R_{\mathrm{main}} = \frac{\epsilon}{1-\epsilon} N^{7/3} |e_U^{(f)}| \,.
$$
In the second inequality above we used scaling to conclude that
\begin{align*}
& \inf\left\{ (1-\epsilon) \frac{3}{5} (6\pi^2)^{2/3} \int_{\R^3} \sigma^{5/3} \,dx  - (1-U)D(\sigma,\sigma):\ \sigma\geq 0\,,\ \int_{\R^3} \sigma \,dx = N\right\} \\
& \qquad = (1-\epsilon)^{-1} N^{7/3} e^{(f)}_U \,.
\end{align*}

Thus, to obtain the claimed lower bound \eqref{eq:main2proofgoal}, it remains to show that $G$ and $\epsilon$ can be chosen such that
$$
\epsilon\mathcal T -\mathcal R_{\mathrm{main}} - \mathcal R_\loc - \mathcal R_{\mathrm{rep}} - \mathcal R_{\mathrm{xc}}
\geq - C (1-U)^2 N^{7/3-2/33} \left( 1 + N^{-40/33} (U/(1-U))^2 \right)
$$
We bound the positive term $\mathcal T$ from below by the Lieb--Thirring inequality \cite{LT},
$$
\mathcal T \geq K \int_{\R^3} \rho_\psi^{5/3} \,dx \,,
$$
for some positive constant $K$. Since, by scaling, $e^{(f)}_U$ is proportional to $-N^{7/3}(1-U)^2$, we have
$$
\mathcal R_{\mathrm{main}} \leq C \epsilon N^{7/3}(1-U)^2 \,.
$$
To bound $\mathcal R_\loc$ and $\mathcal R_{\mathrm{rep}}$ we choose $G(x) = \ell^{-3/2}g(x/\ell)$ with some $\ell>0$ to be determined and find that
$$
\mathcal R_\loc = (1-\epsilon) N\ell^{-2} \|\nabla g\|^2 \leq N\ell^{-2} \|\nabla g\|^2 \,.
$$
Moreover, by Lemma \ref{coulomb}, if $g$ is radial and has support in the unit ball, then
$$
\mathcal R_{\mathrm{int}} \leq C(1-U)\ell^{1/5} \|\rho_\psi\|_1 \|\rho_\psi\|_{5/3} = C(1-U)\ell^{1/5} N \|\rho_\psi\|_{5/3} \,.
$$
Finally, by H\"older,
$$
\mathcal R_{\mathrm{xc}} \leq 1.68\ U \|\rho_\psi\|_1^{1/2} \|\rho_\psi\|_{5/3}^{5/6} = 1.68\ U N^{1/2} \|\rho_\psi\|_{5/3}^{5/6}\,.
$$
In order to balance the errors coming from the localization and the repulsion, we choose $\ell$ proportional to $((1-U)\|\rho_\psi\|_{5/3})^{-5/11}$. To summarize, we have
\begin{align*}
& \epsilon\mathcal T -\mathcal R_{\mathrm{main}} - \mathcal R_\loc - \mathcal R_{\mathrm{rep}} - \mathcal R_{\mathrm{xc}} \\
& \qquad \geq \epsilon K \|\rho_\psi\|_{5/3}^{5/3} - C \left( \epsilon (1-U)^2 N^{7/3} + (1-U)^{10/11} N \|\rho_\psi\|_{5/3}^{10/11} + U N^{1/2} \|\rho_\psi\|_{5/3}^{5/6} \right) \,.
\end{align*}
Minimizing $(\epsilon K/2) \|\rho_\psi\|_{5/3}^{5/3} - C(1-U)^{10/11} N \|\rho_\psi\|_{5/3}^{10/11}$ and $(\epsilon K/2) \|\rho_\psi\|_{5/3}^{5/3} - C U N^{1/2} \|\rho_\psi\|_{5/3}^{5/6}$ with respect to $\|\rho_\psi\|_{5/3}$, we obtain
\begin{align*}
& \epsilon\mathcal T -\mathcal R_{\mathrm{main}} - \mathcal R_\loc - \mathcal R_{\mathrm{rep}} - \mathcal R_{\mathrm{xc}} \\
& \qquad \geq - C \left( \epsilon (1-U)^2 N^{7/3} + \epsilon^{-6/5} (1-U)^2 N^{11/5} + \epsilon^{-1} U^2 N \right) \,.
\end{align*}
Finally, we the choice $\epsilon= N^{-2/33}$ we obtain \eqref{eq:main2proofgoal}. This proves the claimed lower bound, except for the following lemma that was used in the proof.

\begin{lemma}\label{coulomb}
Let $\sigma$ be a non-negative, radially symmetric function with support in a ball of radius $R>0$ and $\int_{\R^3}\sigma\,dx = 1$. Then, for all $\rho\in L^1(\R^3)\cap L^{5/3}(\R^3)$,
$$
0\leq D(\rho,\rho) - D(\rho*\sigma,\rho*\sigma) \leq C R^{1/10} \|\rho\|_1 \|\rho\|_{5/3} 
$$
\end{lemma}

\begin{proof}
The left inequality is easily verified in Fourier space, or by using Newton's theorem, and we concentrate on proving the right one. In fact, we shall show that
$$
D(\rho-\rho*\sigma,\tau) \leq C R^{1/5} \|\rho\|_1 \|\tau\|_{5/3} \,.
$$
Then, writing
$$
D(\rho,\rho) - D(\rho*\sigma,\rho*\sigma) = D(\rho-\rho*\sigma,\rho+\rho*\sigma)
$$
and noting that $\|\rho+ \rho*\sigma\|_{5/3}\leq \|\rho\|_{5/3} +\|\rho*\sigma\|_{5/3}\leq 2 \|\rho\|_{5/3}$, we will obtain the inequality of the lemma.

Thus, it remains to prove the above inequality. By H\"older's and Young's inequality,
$$
2 D(\rho-\rho*\sigma,\tau) \leq \| |x|^{-1}*\rho*\sigma - |x|^{-1}*\rho\|_{5/2} \|\tau\|_{5/3} 
\leq \|\rho\|_1 \||x|^{-1}*\sigma - |x|^{-1} \|_{5/2} \|\tau\|_{5/3} \,.
$$
By Newton's theorem, we have $0\leq |x|^{-1}- |x|^{-1}*\sigma \leq |x|^{-1} \chi_{\{|x|<R\}}$. Thus,
$$
\left\||x|^{-1}*\sigma - |x|^{-1} \right\|_{5/2}^{5/2} \leq \int_{\{|x|<R\}} \frac{dx}{|x|^{5/2}} = 8\pi  R^{1/2} \,.
$$
This proves the claimed inequality.
\end{proof}



\bibliographystyle{amsalpha}

\end{document}